\documentclass[twocolumn]{autart} 

\usepackage{amsmath}
\usepackage{amssymb}
\usepackage{graphicx}
\usepackage{epstopdf}
\usepackage{bm}
\usepackage{color}
\usepackage{cite}
\usepackage{subfigure}
\usepackage{algorithm}
\usepackage{algpseudocode}

\usepackage{booktabs}
\usepackage{siunitx}

\newtheorem{theorem}{Theorem}

\newtheorem{assumption}{Assumption}
\newtheorem{remark}{Remark}

\newenvironment{proof}{\noindent{\bf Proof}\,\,}{$\hfill\blacksquare$\\}

\allowdisplaybreaks

\usepackage{todonotes}

\begin{document}

\begin{frontmatter}

\title{Deep Reinforcement Learning for Wireless Sensor Scheduling in Cyber-Physical Systems\thanksref{footnoteinfo}} 
 
\thanks[footnoteinfo]{A. Ramaswamy was supported by the German Research Foundation (DFG) - 315248657. L. Shi was supported by a Hong Kong RGC General Research Fund
16204218.} 
 
\author[Paderborn]{Alex S. Leong}\ead{alex.leong@upb.de},   
\author[Paderborn]{Arunselvan Ramaswamy}\ead{arunr@mail.uni-paderborn.de},   
\author[Paderborn]{Daniel E. Quevedo}\ead{dquevedo@ieee.org},
\author[Paderborn]{Holger Karl}\ead{h.karl@upb.de},
\author[HKUST]{Ling Shi}\ead{eesling@ust.hk}

\address[Paderborn]{Faculty of Computer Science, Electrical Engineering and Mathematics, Paderborn University, Paderborn, Germany}  
\address[HKUST]{Department of Electrical and Computer Engineering, Hong Kong University of Science and Technology, Hong Kong}   

\maketitle

%

\maketitle

\begin{abstract}
In many  Cyber-Physical Systems, we encounter the problem of remote state estimation of geographically distributed and remote physical processes. This paper studies the scheduling of sensor transmissions  to estimate the states of multiple remote, dynamic processes. Information from the different sensors have to be transmitted to a central gateway over a wireless network for monitoring purposes, where  typically fewer wireless channels are available than there are processes to be monitored.  For effective estimation at the gateway, the sensors need to be scheduled appropriately, i.e., at each time instant one needs to decide which sensors have network access and which ones do not.  To address this scheduling problem, we formulate an associated Markov decision process (MDP). This MDP is then solved using a Deep Q-Network, a recent deep reinforcement learning  algorithm that is at once scalable and model-free. We compare our scheduling algorithm to popular scheduling algorithms such as round-robin and reduced-waiting-time, among others. Our algorithm is shown to significantly outperform these algorithms for many example scenarios.
\end{abstract}

\end{frontmatter}

\section{Introduction}
\label{intro_sec}
Cyber-physical systems (CPS) are systems built through integration of sensors, communication networks, controllers, dynamic (physical) processes and actuators.
They are playing an increasingly important role in
modern society, in areas such as energy, transportation, manufacturing, and
healthcare. The scale of typical CPS such as smart-grids, vehicular traffic networks and smart factories is large.  The realization of these systems faces substantial challenges arising
in diverse disciplines, ranging from communications and control to computing
\cite{special_issue_PROC12}. 
Supporting  estimation and control applications over wireless networks
has posed considerable challenges for the operation of networks and
the design of protocols \cite{special_issue_TAC15}. 

Figure~\ref{system_model} illustrates an example of a networked cyber-physical system for the purposes of remote state estimation. A number of processes are observed by sensors, with the sensors sending information 
via a shared wireless network (consisting of $M$ wireless channels) to a gateway, that computes state estimates of each of these processes. Such situations could, for instance,
occur if a central controller wishes to monitor a number of different
processes in an industrial plant.
 \begin{figure}[t!]
\centering 
\includegraphics[scale=0.4]{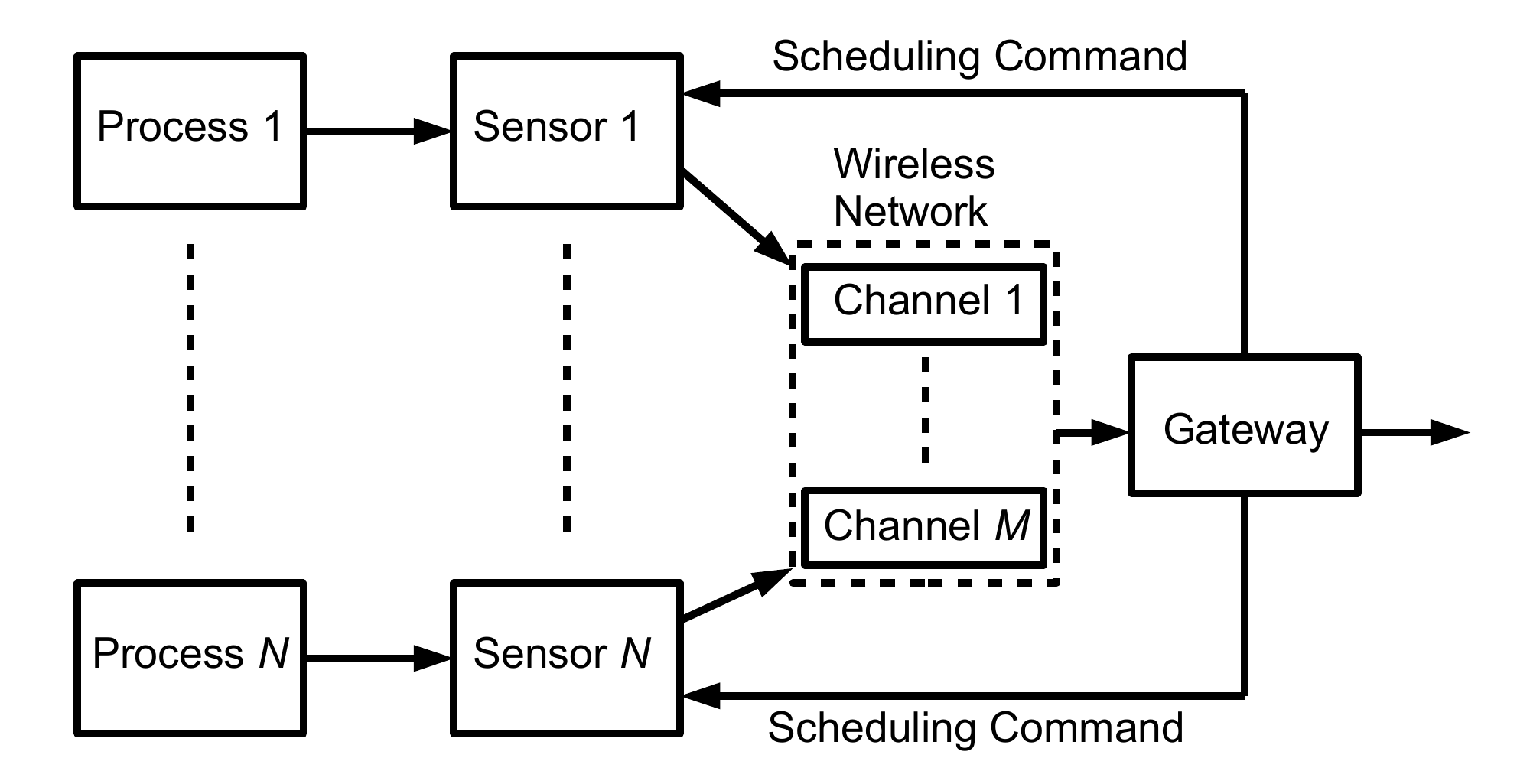} 
\caption{Remote state estimation with sensor scheduling}
\label{system_model}
\end{figure} 
From a networking perspective, one challenge lies in scheduling
transmissions from the sensors to the gateway,
because of both the volatile nature of wireless channels
and the need to carefully schedule transmissions over a shared medium
\cite{Molisch}.
While such channels provide the
opportunity for diversity, they also aggravate the dynamic scheduling problem:
which channel should be assigned to which sensor, and when?
    The problem of scheduling is further exacerbated by estimation and control requirements, which may be at odds with typical communications performance parameters such as waiting times, throughput, etc. \cite{ChaskarMadhow,WuSrikantPerkins}. 

The sensor scheduling problem wherein a single dynamic process is observed
by multiple sensors has been studied in
e.g. \cite{HovareshtiGuptaBaras,MoGaroneSinopoli,ZhaoZhangHu,LeongDeyQuevedo_TAC}. More
recently, sensor scheduling problems where multiple processes are
observed by different sensors has also been investigated
\cite{HanWuZhangShi,WuRenDeyShi}. In the case of single channel systems ($M=1$),  optimal sensor scheduling
problems without  packet drops have been previously studied in \cite{HanWuZhangShi}. For the case $M > 1$  and additionally with packet transmission length
constraints, some structural results were derived in \cite{WuRenDeyShi}, however
numerical results were only provided for the $M=1$ case. The focus of
the current paper is on the case $M >1$, where each wireless channel can also experience packet drops. In particular, we want to provide
computationally scalable methods for solving  optimal sensor
scheduling problems. 



For the dynamic scheduling problem,
 the gateway selects at each discrete time instant a subset (of size $M$) of the $N$ sensors which communicate the
sensor readings to the gateway, 
to update its estimates. 
We assume that the gateway has  knowledge of the process dynamics observed by each sensor, to allow Kalman filter-type estimation algorithms to be run.
The scheduling decision could be informed by knowledge about the
quality of the estimates as well as by conjectures about channel state
and probability of success of transmitting the readings to the
gateway. Knowledge of the channel states or channel statistics is not assumed to be known to the gateway (i.e. scheduling is done in a model-free manner), as such knowledge may be expensive to obtain (requiring e.g. the transmission of pilot signals), and furthermore since channel statistics are often also time-varying \cite{Eisen_ACC}.

As previously mentioned, the scale of a CPS is typically large. For our scheduling problem, this leads to an associated MDP with large state and action spaces. Traditional reinforcement learning based algorithms such as $Q$-learning  cannot be used to solve such MDPs due to Bellman's curse of dimensionality \cite{Bertsekas_DP1}. The curse of dimensionality can be overcome by the use of function approximations \cite{SuttonBarto}. Deep Q-Network (DQN) \cite{Mnih_Atari,Mnih_nature} is one such algorithm using deep neural networks as function approximators, that has shown tremendous promise in solving large MDPs in a scalable, model-free manner.  Deep reinforcement learning techniques have also been recently used to study difficult problems arising in control. The work  \cite{DemirelRamaswamy} studies a similar problem in controller scheduling, however it does not consider packet drops, and requires extra overhead in the transmission of  information from the sensors to the scheduler at every time step. The work of \cite{BaumannTrimpe} studies event-triggered control problems where the communication and control policies are learnt from scratch using an actor-critic approach.

The paper is organized as follows. The system model is presented in Section \ref{model_sec}. The sensor scheduling problem and associated MDP is described in Section \ref{optimal_scheduling_sec}, together with derivation of a stability condition and discussion of computational issues. The proposed deep reinforcement learning approach to the  scheduling problem is given in Section~\ref{deep_RL_sec}. Numerical studies can be found in Section \ref{numerical_sec}.

\section{System Model}
\label{model_sec}

\subsection{Sensing model}
\label{sec:sensing-model}

A diagram of the system model is shown in Fig. \ref{system_model}.
We consider $N$ independent, linear, discrete-time processes
\begin{equation}
\label{state_eqn}
x_{i,k+1} = A_i x_{i,k} + w_{i,k}, \quad i = 1,\dots,N
\end{equation}
where $x_{i,k} \in \mathbb{R}^{n_{x_i}}$ is the state of process $i$ at time $k$, and the process noise  $w_{i,k}$ is i.i.d. (in time) Gaussian with zero mean and covariance matrix $W_i \geq 0$.\footnote{For a symmetric matrix $X$, we say that $X> 0 $ if it is positive definite, and  $X \geq 0 $ if it is positive semi-definite.}
Each process is measured by a  sensor as
\begin{equation}
\label{measurement_eqn}
y_{i,k} = C_i x_{i,k} + v_{i,k}, \quad i=1,\dots,N
\end{equation}
where $y_{i,k} \in \mathbb{R}^{n_{y_i}}$ is the measurement  of process $i$ at time $k$, and the measurement noise $v_{i,k}$ is i.i.d. Gaussian with zero mean and covariance matrix $V_i > 0$. 
The noise processes $\{w_{i,k}\}$ and $\{v_{j,k}\}$ are assumed to be
mutually independent for all $i$ and $j$. 

We assume that each sensor has the computational capability to run a Kalman filter, i.e., each sensor $i$ can compute  local state estimates\footnote{In situations where channels experience packet drops, transmission of local state estimates  in general gives better estimation performance  than transmission of raw measurements \cite{XuHespanha}. It is worth noting that the situation where raw measurements are transmitted  can also be handled using the deep $Q$-learning technique considered in the present work.} and estimation error covariance matrices
\begin{equation*}
\begin{split}
\hat{x}_{i,k|k-1}^s & \triangleq \mathbb{E}[x_{i,k}|y_{i,0},\dots,y_{i,k-1}]  \\
\hat{x}_{i,k}^s & \triangleq \mathbb{E}[x_{i,k}|y_{i,0},\dots,y_{i,k}] \\
P_{i,k|k-1}^s & \triangleq  \mathbb{E}[(x_{i,k}-\hat{x}_{i,k|k-1}^s)(x_{i,k}-\hat{x}_{i,k|k-1}^s)^T  \\ & \quad\quad\quad |y_{i,0},\dots,y_{i,k-1}]\\
P_{i,k}^s & \triangleq  \mathbb{E}[(x_{i,k}-\hat{x}_{i,k}^s)(x_{i,k}-\hat{x}_{i,k}^s)^T|y_{i,0},\dots,y_{i,k}],
\end{split}
\end{equation*} 
using the Kalman filter equations \cite{AndersonMoore}.
We will assume that every  pair $(A_i,C_i)$ is observable, and every pair
$(A_i, W_i^{1/2})$  is controllable. Then, the steady-state value of $P_{i,k}^s$ for $k \rightarrow \infty$ exists for each sensor, and will be denoted by $\overline{P}_i$. For convenience of presentation, we will assume that the local Kalman filters at the sensors have reached steady state\footnote{Convergence to steady state in general occurs at an exponential rate \cite{AndersonMoore}.}, so that $P_{i,k}^s = \overline{P}_i, \forall i=1,\dots,N, \forall k$. 

\subsection{Scheduling and channel model}
\label{sec:sched-chann-model}

The sensors wish to transmit their local state estimates
$\hat{x}_{i,k}^s$ to a central gateway, which aims to estimate all of
the $N$ processes $\{x_{i,k}\}, i=1,\dots,N$. Sensor transmissions are over a
shared wireless network with $M$ channels. In typical applications, $
M \ll N$ due to limited resources. Thus, (at most) only $M$ out of the $N$
sensors can transmit at any given time. At each time step $k$, a scheduler will  allocate each of the $M$ channels to one of the sensors. We assume that each channel is allocated to a different sensor, although the case where multiple channels are allocated to the same sensor (e.g. as in \cite{MesquitaHespanhaNair}) can also be handled using our techniques. Define decision variables $a_{m,k} \in \{1,\dots,N\}$ for $m = 1,\dots,M$  as 
\begin{align}
a_{m,k} \triangleq   i & \textnormal{ if sensor $i$ is scheduled to transmit on} \nonumber \\ & \textnormal{ channel $m$ at time $k$}. \label{a_mk_defn}   
\end{align}

Channel transmissions can experience packet drops. Define $\gamma_{m,k} \in \{0,1\}$ for $ m = 1,\dots,M$ such that
$$\gamma_{m,k} \triangleq  \left\{\begin{array}{cl} 1, & \textnormal{if transmission on channel $m$ at time $k$} \\ & \textnormal{is successfully received at gateway} \\ 0, & \textnormal{otherwise}.  \end{array}  \right. $$
Each  channel is modelled using the Gilbert-Elliott (or Markovian packet drop \cite{HuangDey}) model, with 
\begin{align*} 
& p_m \triangleq \mathbb{P} (\gamma_{m,k} = 0 | \gamma_{m,k-1} = 1), \\ & q_m \triangleq \mathbb{P} (\gamma_{m,k} = 1 | \gamma_{m,k-1} = 0), \quad m = 1,\dots,M,
\end{align*}
and with the channels being independent of each other. $p_m$ and $q_m$ are also known  respectively as the failure rate and recovery rate.  As mentioned in the Introduction, we will not assume knowledge of the channel parameters $p_m, q_m, m=1,\dots,M$ at the scheduler.  We note that our model-free approach can also be readily extended to handle more general finite state Markov channels \cite{SadeghiKennedy,QuevedoOstergaardAhlen}.

\subsection{Protocol assumptions}
\label{sec:protocol-assumptions}

Scheduling is assumed to be done  at the gateway, with the decisions
$a_{m,k}$ fed back to the sensors.\footnote{Scheduling can also be
  done inside the network (e.g., at a wireless access point) provided
  $\gamma_{m,k-1}$ are fed back to the network to allow $P_{i,k-1},
  i=1,\dots,N$ to be reconstructed. This makes no difference for the
  approach considered here.}
We assume
that this (downlink) transmission from gateway to sensor works
without errors. We justify this by using all $M$
stochastically independent channels to transmit this signalling
information, resulting in an exponentially reduced error
probability. 
Error performance can be further improved by coding across channels
(rather than just simple repetition coding) and time (since signalling
information is relatively small, time overhead can be invested) \cite{Proakis,Molisch}.

After these channel assignments have been received by the sensors,
they send their respective data (local state estimates)  to the gateway.
Once these (uplink) transmissions
are complete, we move to the next time period~$k+1$.

\subsection{Remote Estimation at Gateway}
At the gateway, state estimates and  estimation error covariances of each of the processes are computed similar to \cite{XuHespanha,ShiEpsteinMurray}, as follows:
\begin{equation}
\label{remote_estimator_eqns_multi_sensor}
\begin{split}
\hat{x}_{i,k} & = \left\{\begin{array}{cl} \hat{x}_{i,k}^s, & \textnormal{if $\exists m$ s.t. $a_{m,k} = i$ and $\gamma_{m,k} = 1$}   \\ A_i \hat{x}_{i,k-1}, & \textnormal{otherwise}  \end{array}  \right. \\
P_{i,k} & = \left\{\begin{array}{cl}  \overline{P}_i, &   \textnormal{if $\exists m$ s.t. $a_{m,k} = i$ and $\gamma_{m,k} = 1$}  \\ h_i(P_{i,k-1}),  & \textnormal{otherwise},   \end{array} \right. 
\end{split}
\end{equation} 
where  $h_i(.), i=1,\dots,N$, is defined as
\begin{equation}
\label{h_defn}
h_i(X) \triangleq A_i X A_i^T + W_i.
\end{equation}
As mentioned in the Introduction, the gateway is assumed to have knowledge of the parameters for each of the $N$ processes, which allows (\ref{remote_estimator_eqns_multi_sensor}) to be (causally) computed for each process. 


\section{Problem Description}
\label{optimal_scheduling_sec}
The gateway wishes to find a scheduling policy to minimize  the average sum of the trace
of the estimation error covariance matrices across all sensors and all times. We will formulate a Markov decision process (MDP) to solve the associated sequential decision making problem:
\begin{equation}
\label{sensor_scheduling_MDP}
\begin{split}
&\min_{\{(a_{1,k},\dots,a_{M,k})\}} \limsup_{T\rightarrow\infty} \frac{1}{T} \mathbb{E} \left[ \sum_{k=0}^{T-1} \sum_{i=1}^N \textnormal{tr} P_{i,k} \right].
\end{split}
\end{equation}


 We assume that the channel allocations at time $k$ can depend on 
\begin{equation}
\label{MDP_state_original} 
 (P_{1,k-1}, \dots, P_{N,k-1},\gamma_{1,k-1},\dots,\gamma_{M,k-1}),
 \end{equation}
 namely the estimation error covariances and channel transmission outcomes at the previous time step, which is information that is available to the gateway.  From (\ref{remote_estimator_eqns_multi_sensor}) we see that $P_{i,k}$ is always of the form $h^n_i(\overline{P}_i)$ for some $n \in \mathbb{N}$, where $h^n_i(.)$ denotes the $n$-fold composition of $h_i(.)$ given in (\ref{h_defn}), with $h_i^0(.)$ being the identity. Define the holding time of sensor $i$ at time $k$ as
\begin{align*}
\tau_{i,k} \triangleq \min\{\tau \geq 0: & \textnormal{ $\exists m$ s.t. $a_{m,k-\tau} = i$ and $\gamma_{m,k-\tau} = 1$}  \},
\end{align*}
 which represents the amount of time since the last successful transmission of sensor $i$ to the gateway. Then we can express $P_{i,k}$ as 
$$P_{i,k} = h_i^{\tau_{i,k}} (\overline{P}_i),$$
and therefore the channel allocations at time $k$ can, equivalently, depend on 
\begin{equation}
\label{MDP_state_new}
(\tau_{1,k-1}, \dots, \tau_{N,k-1},\gamma_{1,k-1},\dots,\gamma_{M,k-1}),
\end{equation}
which is of smaller dimension than (\ref{MDP_state_original}), as each $\tau_{i,k-1}$ is scalar while each $P_{i,k-1}$ is a matrix. 
 Below we will describe more formally problem (\ref{sensor_scheduling_MDP}) as an MDP. 
 
\subsection{Formulation as a Markov Decision Process}
\label{MDP_formulation_sec}
\emph{State space}: 
 From the discussion above, the vector (\ref{MDP_state_new})
 can  be regarded as the state\footnote{Note that the state of the MDP is different from the states $x_{i,k}$ of the processes. From now on we will mostly use the word ``state'' to refer to the state of an MDP.} of the MDP (\ref{sensor_scheduling_MDP}) at time $k$, and thus the state space is $\mathbb{N}^N \times \{0,1\}^M$ (where we include 0 in the natural numbers $\mathbb{N}$). 

\emph{Action space}:
Next, we have a finite action space 
$$\{(a_{1,k},\dots,a_{M,k}) |   a_{1,k},\dots,a_{M,k} \textnormal{ all distinct} \},$$
corresponding to the  $\frac{N!}{(N-M)!}$ different ways of allocating the $M$ channels to the $N$ sensors. 

\emph{Cost function}:
Finally, the single stage cost at time $k$ is 
\begin{equation}
\label{per_stage_cost}
J_k = \sum_{i=1}^N \textrm{tr} P_{i,k}. 
\end{equation}

\begin{remark}
As the channel parameters $p_m, q_m, m=1,\dots,M$ are assumed to be unknown, we do not include the transition probabilities in our formulation of the MDP, and indeed their knowledge is not required when solving the MDP using reinforcement learning methods.
\end{remark}

\subsection{Stability Condition}
\label{stability_sec}
We will derive a sufficient condition on when the optimal solution to the MDP (\ref{sensor_scheduling_MDP}) has  bounded average cost, expressed in terms of the process and channel parameters. Such a stability condition is important for reliable monitoring of all of the processes.  We first make the following assumption: 

\begin{assumption}
\label{stability_assumption}
Define $\rho_{\textnormal{max}} \triangleq \max_{i=1,\dots,N} \rho(A_i)$ and $q_{\textnormal{max}} \triangleq \max_{m=1,\dots,M} q_m$, where $\rho(A_i)$ denotes the spectral radius of $A_i$. We assume that 
\begin{equation}
\label{stability_condition}
 \rho_{\textnormal{max}}^2 (1-q_{\textnormal{max}}) < 1.
 \end{equation}
\end{assumption}
\begin{theorem}\label{thm_stability}
Under Assumption \ref{stability_assumption}, the optimal solution to the MDP (\ref{sensor_scheduling_MDP}) has  bounded average cost.
\end{theorem}

\begin{proof}
See the appendix.
\end{proof}

\begin{remark}
For the case of a single process and a single Gilbert-Elliott channel (with transition parameters $p$ and $q$),  when local state estimates are transmitted, a necessary and sufficient condition for bounded expected estimation error covariance is that $q$ satisfies \cite{GuptaHassibiMurray}:
\begin{equation}
\label{single_sensor_stability_condition}
\rho(A)^2 (1-q) < 1.
\end{equation}
The condition (\ref{stability_condition}) can be regarded as a generalization of  (\ref{single_sensor_stability_condition}) to multiple processes and multiple channels, and intuitively says that the overall system has bounded cost provided the best channel (in terms of having the largest recovery rate $q_m$) can keep the expected estimation error covariance of the most unstable process (i.e., having the largest spectral radius)  bounded. 
\end{remark}

\subsection{Computational Issues}
Considering first the case where the channel parameters $p_m, q_m, m=1,\dots,M$ are known, numerical solution of (\ref{sensor_scheduling_MDP}) using dynamic programming techniques (e.g. using policy iteration or relative value iteration) is in principle possible, after truncating the countable state space $\mathbb{N}^N \times \{0,1\}^M$   to a finite state space. However in practice, even for relatively small $N$ and $M$, the sizes of both the state and action spaces can still be considerable, making exact numerical solution infeasible. For the case $M=1$ without packet drops (and relatively small $N$ in numerical computation), a similar average cost problem has been previously studied \cite{HanWuZhangShi}. For  $M > 1$  and additionally also considering packet transmission length constraints, some structural results were derived in \cite{WuRenDeyShi}, however numerical results were only provided for the $M=1$ case. 

If the channel parameters $p_m, q_m, m=1,\dots,M$, are unknown (and hence the MDP transition probabilities are also unknown), as is assumed in the current work, then standard dynamic programming approaches for solving MDPs cannot be used. 

In order to overcome the above mentioned problems of large state space and unknown channel parameters, we will use recently developed reinforcement learning ($Q$-learning) methods utilizing deep neural networks for function approximation \cite{Mnih_Atari,Mnih_nature}, which will be described in the next section.


\section{Sensor Scheduling Using Deep Reinforcement Learning}
\label{deep_RL_sec}
Consider  the discounted cost problem 
\begin{equation}
\label{sensor_scheduling_MDP_discounted_cost}
\begin{split}
& \min_{\{(a_{1,k},\dots,a_{M,k})\}} \limsup_{T\rightarrow\infty} \mathbb{E} \left[ \sum_{k=0}^{T-1} \sum_{i=1}^N  \delta^k \textnormal{tr} P_{i,k} \right] 
\end{split}
\end{equation}
where $\delta < 1 $ is a discount factor. 
In this paper we will approximate the solution to problem (\ref{sensor_scheduling_MDP}) by solving (\ref{sensor_scheduling_MDP_discounted_cost}) using reinforcement learning techniques, with a discount factor $\delta$ close to 1 \cite{HernandezLermaLasserre}. While $Q$-learning type algorithms for average reward maximization problems exist \cite{Bertsekas_DP2,AbounadiBertsekasBorkar}, most reinforcement learning algorithms assume a discounted setting, in particular the  deep reinforcement learning techniques of \cite{Mnih_Atari,Mnih_nature}.  A more formal justification for solving the discounted cost problem  will be given in Section~\ref{discount_avg_relation_sec}.

\subsection{Solving the discounted cost problem using deep reinforcement learning}
Let us rewrite (\ref{sensor_scheduling_MDP_discounted_cost}) as the equivalent discounted reward maximization problem:
\begin{equation}
\label{sensor_scheduling_MDP_discounted_reward}
\begin{split}
& \max_{\{(a_{1,k},\dots,a_{M,k})\}} \liminf_{T\rightarrow\infty} \mathbb{E} \left[ \sum_{k=0}^{T-1} \sum_{i=1}^N  - \delta^k \textnormal{tr} P_{i,k} \right]. 
\end{split}
\end{equation}

The $Q$-factor or action-value function
  $Q(s,a)$ represents the expected future reward associated with taking action $a$ when at state  $s$ \cite{Bertsekas_DP2,SuttonBarto}. 
The $Q$-factor version of the Bellman equation  for problem (\ref{sensor_scheduling_MDP_discounted_reward}) is:
$$Q^*(s,a) = \mathbb{E} \left[ r + \delta \max_{a'} Q^*(s',a') | s,a \right],$$
where $s'$ represents the value of the next state given the current state $s$ and action $a$, and $Q^*(.,.)$ are the optimal $Q$-factors. 
If we know $Q^*(.,.)$, then  we can find a corresponding optimal stationary policy, with action $a^*(s)$ for each state $s$ as follows:
$$a^*(s) = \textrm{argmax}_{a} Q^*(s,a). $$ 

The well-known $Q$-learning algorithm  will, in principle, converge to  the optimal $Q$-factors, but in practice the convergence is rather slow and requires both the state and action spaces to be small in order for the method to be feasible. 
For large MDPs one can approximate $Q^*(s,a)$ by a  function $Q(s,a;\theta)$ parameterized by a set of weights $\theta$ \cite{SuttonBarto}, and then learning these weights. Deep reinforcement learning refers to the case where the function approximation $Q(s,a;\theta)$ uses a (deep) neural network, which has been crucial in recent key breakthroughs in artificial intelligence such as in the playing of Go \cite{Silver_alphago}. The  deep $Q$-learning techniques introduced in \cite{Mnih_Atari,Mnih_nature} also included a number of important innovations aimed at stabilizing the learning algorithm, in particular 1) the notion of experience replay\footnote{In experience replay we store the agent's experiences at each time-step, pooled over many episodes, into a replay memory. During the minibatch updates, random samples from the replay memory are drawn. Such a technique can reduce correlations in the observation data.}  (see step 9 of Algorithm \ref{alg:deepQscheduling} below), and 2) fixing the target $Q$-network at regular intervals\footnote{This technique can reduce correlations between the $Q$-factors and the target.} (see step~12 of Algorithm \ref{alg:deepQscheduling}). 
Based on these ideas, our approach to solving problem (\ref{sensor_scheduling_MDP_discounted_reward}) is given as Algorithm \ref{alg:deepQscheduling} below.

In Algorithm\,\ref{alg:deepQscheduling}, 
$$a_t = (a_{1,t},\dots,a_{M,t}),$$
c.f. (\ref{a_mk_defn}), 
corresponds to the allocation of the $M$ channels  at time $t$, and the single stage reward is given by
$$r_t = \sum_{i=1}^N - \textrm{tr} P_{i,t}.$$  The state $s_t$ could be chosen as 
$$s_t = (\tau_{1,t-1},\dots,\tau_{N,t-1},\gamma_{1,t-1},\dots,\gamma_{M,t-1})$$ as in Section \ref{MDP_formulation_sec}, however for the simulations in Section~\ref{numerical_sec} we  further augment the state to 
\begin{align}
\label{augmented_state}
s_t = \big(&\tau_{1,t-1},\dots,\tau_{N,t-1},\textnormal{tr} (h_1(P_{1,t-1})), \dots,  \nonumber \\ & \quad \textnormal{tr} ( h_N (P_{N,t-1})),\gamma_{1,t-1},\dots,\gamma_{M,t-1}\big),
\end{align}
 where $\textnormal{tr} (h_i(P_{i,t-1}))$ is directly related to the reward function at time $t$ when we don't receive transmission from sensor $i$, 
which we have found in some cases gives faster convergence for the algorithm. For details of the hyper-parameters for Algorithm \ref{alg:deepQscheduling} used in this paper, see Section \ref{numerical_sec}.
We note that Algorithm\,\ref{alg:deepQscheduling} can be run online, and is model-free in that it does not need knowledge of the channel parameters $p_m, q_m, m=1,\dots,M$. 

\begin{algorithm}[t]
\caption{Deep $Q$-network for wireless sensor scheduling}
\label{alg:deepQscheduling}
\begin{algorithmic}[1]
\State Initialize replay memory $\mathcal{D}$ to capacity $K$
\State Initialize network $Q$ with random weights $\theta_0$
\State Initialize target network $\hat{Q}$ with weights $\theta^- = \theta_0$
	\State Initialize $s_0 $
	\For{$t=0,1,\dots,T$}
		\State With probability $\varepsilon$ select a random action $a_t$, otherwise select $a_t = \textnormal{argmax}_a Q(s_t,a;\theta_t)$
		\State Execute $a_t$, and observe $r_t$ and $s_{t+1}$
		\State Store $(s_t,a_t,r_t,s_{t+1})$ in $\mathcal{D}$
		\State Sample random mini-batch of transitions $(s_j,a_j,r_j,s_{j+1})$ from $\mathcal{D}$
		\State Set $z_j = r_j + \delta \max_{a'} \hat{Q}(s_{j+1},a';\theta^-)$ for each sample in mini-batch
		\State Perform a mini-batch gradient descent step on $(z_j - Q(s_j,a_j; \theta_t))^2$  to obtain $\theta_{t+1}$
		\State Every $c$ steps set $\theta^- = \theta_t$
	\EndFor
\end{algorithmic}
\end{algorithm}

\subsection{Relationship to average cost problem}
\label{discount_avg_relation_sec}
As stated in Section~\ref{optimal_scheduling_sec}, the aim of the scheduler is to find a scheduling policy that minimizes the average estimation error covariances, i.e., solves an associated average cost problem. If the communication channels satisfy Assumption \ref{stability_assumption}, then it follows from Theorem~\ref{thm_stability} that there exists a  scheduling policy that ensures that the cost is bounded. 
In this subsection, we show that the policy found by solving the associated discounted cost problem  is an $\epsilon$-optimal policy for the average cost problem.\footnote{Note that $\epsilon$ here is different from the exploration parameter $\varepsilon$ of Algorithm 1.} Further, $\epsilon$ can be made arbitrarily small by controlling the discount factor, $\delta$, of the associated MDP. 

Recall that $J_k$ given by (\ref{per_stage_cost}) is  the single stage cost associated with  problem (\ref{sensor_scheduling_MDP}). Before proceeding, we state Abel's theorem \cite{HernandezLermaLasserre} for our setting:
\begin{theorem}[Abel]
 Let $\{J_k\}_{k \ge 0}$ be a sequence of positive real numbers. Then
  \begin{align*}
&  \liminf \limits_{T \to \infty} \frac{1}{T} \sum \limits_{k=0}^{T-1} J_k  \le \liminf \limits_{\delta \uparrow 1} (1 - \delta) \sum \limits_{k=0}^\infty \delta ^k J_k \\ & \quad \le \limsup \limits_{\delta \uparrow 1} (1 - \delta) \sum \limits_{k=0}^\infty \delta ^k J_k \le \limsup \limits_{T \to \infty} \frac{1}{T} \sum \limits_{k=0}^{T-1} J_k.
 \end{align*}
\end{theorem}
From Theorem~\ref{thm_stability} it follows that there exist (stabilizing) scheduling policies with finite associated average costs. It now follows from Abel's theorem that:
\begin{equation}
 \label{eq_avg_vanish_equal}
  \lim \limits_{T \to \infty} \frac{1}{T} \sum \limits_{k=0}^{T-1} J_k = \lim \limits_{\delta \uparrow 1} (1 - \delta) \sum \limits_{k=0}^\infty \delta ^k J_k < \infty.
\end{equation}
Furthermore, given $\epsilon > 0$, there exists an $\delta(\epsilon) \approx 1$, dependent on $\epsilon$, such that:
\begin{align*}
 \lim \limits_{\delta \uparrow 1} (1 - \delta) \sum \limits_{k=0}^\infty \delta ^k J_k &\le (1 - \delta(\epsilon)) \sum \limits_{k=0}^\infty \delta(\epsilon) ^k J_k + \epsilon,\\
 \Rightarrow \lim \limits_{T \to \infty} \frac{1}{T} \sum \limits_{k=0}^{T-1} J_k &\le (1 - \delta(\epsilon)) \sum \limits_{k=0}^\infty \delta(\epsilon) ^k J_k + \epsilon. 
\end{align*}
In addition to $\epsilon$, $\delta(\epsilon)$ also depends on the actual realizations of the single stage cost sequences $\{J_k\}_{k \ge 0}$. If one wishes to find an $\epsilon$-optimal policy, then one can choose a discount factor $\delta(\epsilon)$, provided the ``orders'' of these single stage costs are known. In our problem, the single stage costs are unbounded. However, it is clear that the discount factor $\delta \uparrow
1$  as $\epsilon \downarrow 0$. Hence, in our numerical experiments, we choose a discount factor  close to $1$. 



\section{Numerical Studies}
\label{numerical_sec}
We consider an example with $N = 6$ sensors and $M = 3 $ channels. Each process has state dimension 2 (i.e. $n_{x_i} = 2, i = 1,\dots,N$) and scalar measurements ($n_{y_i} = 1, i = 1,\dots,N$). The process parameters $A_i, C_i, W_i, V_i, i=1,\dots,N$  and channel transition probabilities $p_m, q_m, m=1,\dots,M$ are randomly generated. The eigenvalues of $A_i$ are drawn uniformly from the range $(0, 1.3)$. The entries of $C_i$ are drawn uniformly from the range $(0,1)$, and  $W_i$ and $V_i$ are generated by random orthogonal transformations of a diagonal matrix with random diagonal entries  drawn uniformly from the range $(0.2,1.0)$. The channel transition probabilities $p_m$ and $q_m$ are uniformly generated from the range $(0,1)$. 

The following hyper-parameters for Algorithm\,\ref{alg:deepQscheduling} are used in our simulations. 
In the deep-$Q$ network, the augmented state (\ref{augmented_state}) of dimension $2N+M$ is fed in as input, i.e. there is an input layer with $2N+M=15$ nodes.  We use two hidden layers, with each hidden
layer having 1024 nodes, and a fully connected layer with outputs for
each of the $N!/(N-M)!  = 120$ actions. The discount factor is set to
$\delta = 0.95$. The experience replay memory has size $K =
20000$. The exploration parameter $\varepsilon$ in step 6 of Algorithm \ref{alg:deepQscheduling} is attenuated from 1 to
0.01 at the rate of 0.999, i.e. $\varepsilon \leftarrow \max(0.999\varepsilon,0.01)$ after every iteration.
In the neural network training (step 11 of Algorithm~\ref{alg:deepQscheduling}) the ADAM optimizer \cite{KingmaBa} is used with an initial learning rate of $e^{-4}$ and a learning rate decay of 0.001.\footnote{If $\alpha_t$ represents the learning rate at iteration $t$,  $\alpha_0$ the initial learning rate, and $d$ the decay, then $\alpha_t = \frac{\alpha_0}{1+ d t}$.} The size of each mini-batch is 32. The target $Q$-network is updated once every $c = 100$ time steps. 

Algorithm\,\ref{alg:deepQscheduling} is run to train our deep $Q$-network. In order to get a better idea of the training quality over time, we will reset the process after each $T = 500$, which we will refer to as an episode \cite{SuttonBarto}. Running on a standard Intel Core i7 4790 with 8 Gb RAM (without  GPU), each episode of training when using the above hyper-parameters took around 30 seconds to complete. The empirical average cost
$$ \frac{1}{T} \sum_{k=0}^{T-1} \sum_{i=1}^N  \textnormal{tr} P_{i,k}$$ over different episodes for one randomly generated set of parameters is plotted in Fig. \ref{learning_plot1}. 
\begin{figure}[t!]
\centering 
\includegraphics[scale=0.5]{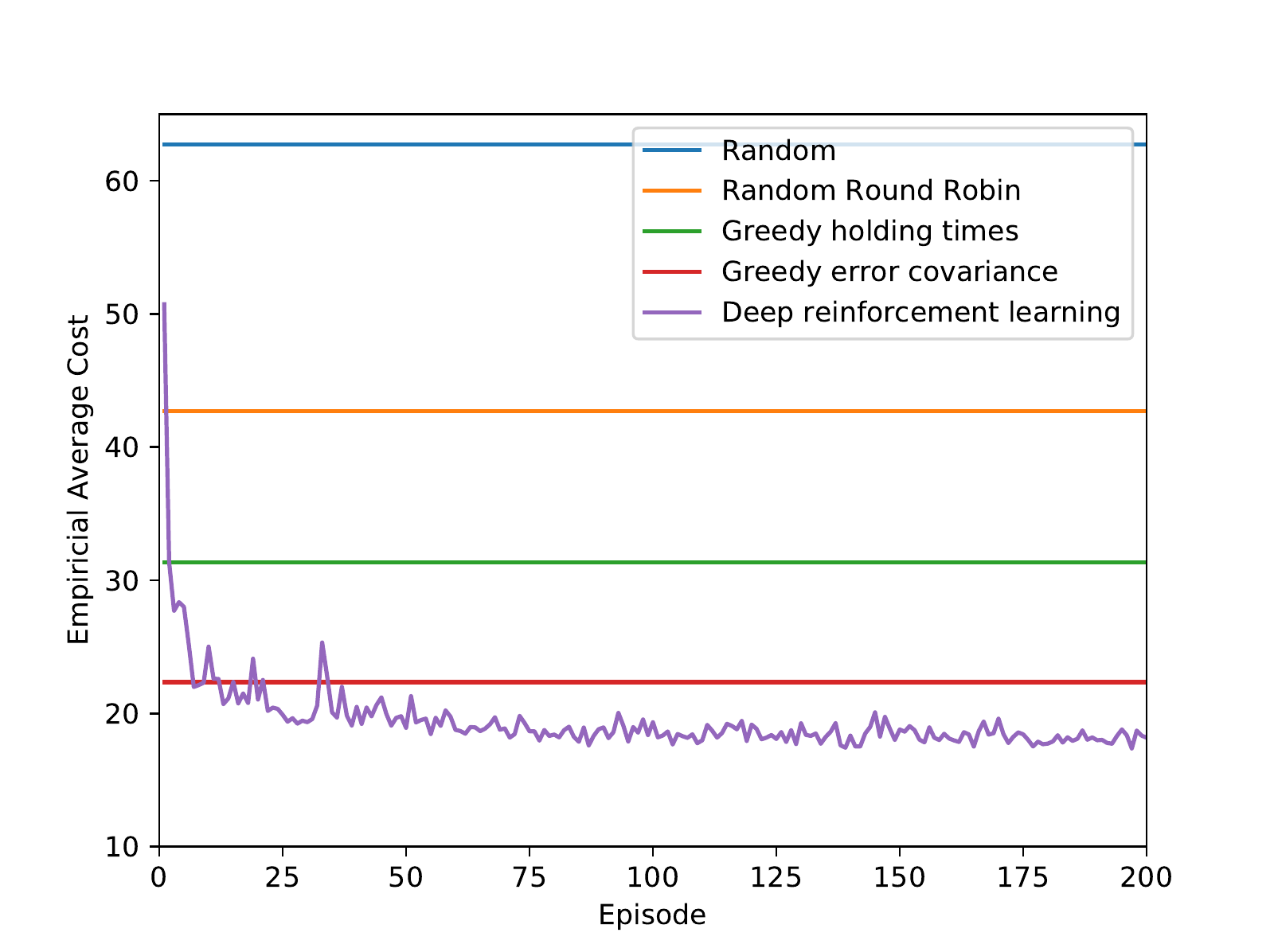} 
\caption{Empirical average cost over different training episodes. The long term average performances of other suboptimal algorithms are also shown for comparison.}
\label{learning_plot1}
\end{figure}  

 We stopped  training after 200 episodes. We then use the trained $Q(.,.; \theta)$ to generate a policy according to
$$a^*(s) = \textrm{argmax}_{a} Q(s,a; \theta). $$
Using the trained policy, simulating the process over 50000 time steps then gives an empirical average cost of around 17.8. 
We compare this performance with the following  policies: 
\begin{enumerate}
\item A random policy that at each time $k$ randomly allocates $M$ out  of the $N$ sensors to the $M$ channels.
\item A round robin policy where $M$ successive sensors (modulo $N$) are randomly allocated to  the $M$ channels at every time instance.\footnote{Round robin schedules are similar to periodic schedules commonly studied in the control literature when there are no packet drops \cite{MoGaroneSinopoli,ZhaoZhangHu}.} 
\item A greedy policy on the holding times, where at each time $k$ we allocate the $M$ sensors with the largest $\tau_{i,k-1}$ (in the case of ties we take the sensors with smallest indices) randomly to the $M$ channels.
\item  A greedy policy on the error covariance, where at each time $k$ we allocate the $M$ sensors with the largest $\textrm{tr} P_{i,k-1}$ randomly to the $M$ channels. 
\end{enumerate}
Simulation over 50000 time steps gives an empirical average cost of around 62.7 for the random policy, 42.7 for the round robin policy, 31.3 for  the greedy policy on holding times, and 22.4 for the greedy policy on error covariances. The performances of these policies are also shown in Fig. \ref{learning_plot1} for comparison. We see that our deep reinforcement learning approach consistently outperforms these policies after around 40-50 episodes of training.

\begin{table*}[t!]
\caption{Empirical average costs for 10 randomly generated sets of parameters}
\centering
\begin{tabular}{rSSSSSS} 
\toprule 
\multicolumn{1}{l}{Param. Set} &  \multicolumn{1}{l}{Random} 
&\multicolumn{1}{l}{Round Robin} 
& \multicolumn{1}{l}{Greedy holding time} 
& \multicolumn{1}{l}{Greedy error covariance} 
& \multicolumn{1}{l}{Deep RL}
& \multicolumn{1}{l}{No replay,}
\\
\multicolumn{1}{l}{} &  \multicolumn{1}{l}{} 
&\multicolumn{1}{l}{} 
& \multicolumn{1}{l}{} 
& \multicolumn{1}{l}{} 
& \multicolumn{1}{l}{}
& \multicolumn{1}{l}{no target $Q$}
  \\ \midrule 
1 & 29151 & 954 & 55.7 & 26.2   & 21.5 & 22.1 \\ 
2 & 1612 & 415 & 80.8 & 49.4  & 36.4  & 41.2\\ 
3 & 2358 & 722 & 80.4 & 51.7   & 32.8  & 44.3\\
4 & 136 & 82.7 & 47.4 & 39.9   & 34.3  & 36.7\\ 
5 & 102 & 42.8 & 17.1 & 13.5   & 10.4  & 10.6 \\ 
6 & 119 &  34.9 & 19.3 & 18.1   & 15.7  & 16.8\\ 
7 & 10097 & 2576 & 58.4 & 42.1   & 35.8  & 39.5\\ 
8 &  65630  & 12555 & 136 & 77.4   & 28.7 & 29.3 \\ 
9 &  37.2  & 30.7 & 25.9 & 23.2   & 21.8 & 22.5\\ 
10 &  29321  & 9049 & 99.4 & 64.6   & 36.7 & 37.7
  \\  \bottomrule
\end{tabular}
\label{comparison_table}
\end{table*}


In Table \ref{comparison_table} we report further comparisons between the random policy, round robin policy, greedy policies, and the performance using deep reinforcement learning, for 10 different randomly generated sets of parameters  $A_i, C_i, W_i, V_i, p_m, q_m, i=1,\dots,N, m=1,\dots,M$ (making sure that condition (\ref{stability_condition}) is satisfied), while keeping $N=6$ and $M=3$.   The same hyper-parameters for training the deep $Q$-network as in the above was used. We can see that the random policy and round robin policy generally do not perform well (although the performance of the round robin policy seems to be better than the purely random policy), and in fact appears to lead to instability in some of the scenarios. The greedy policy on the error covariances performs better than the greedy policy on the holding times, due to the use of more knowledge of the system parameters.  We also see that in each scenario  the approach using deep reinforcement learning performs significantly better than all the other considered policies.  The last column of Table \ref{comparison_table} gives the performance when the techniques from \cite{Mnih_Atari,Mnih_nature} of experience replay and fixing the target $Q$-network are not used. We see that without using these techniques, while in some cases the performance is similar, in other cases there is a significant performance loss.

\begin{remark}
Existing non-control aware scheduling strategies include random, round robin, or greedy strategies with respect to a given parameter, which are also used to, e.g., reduce waiting/holding times. However, in estimation and control applications such strategies do not perform as well as strategies which take into account the dynamics of the processes, as can be seen in Table~\ref{comparison_table}.  
\end{remark}

\section{Conclusion}
\label{sec:conclusions}

This paper has studied a sensor scheduling problem for allocating wireless channels to sensors, for the purposes of remote state estimation of multiple dynamical systems. With the aim of providing a method which can handle larger problems than previous work in the literature, we have proposed an approach based on modern deep reinforcement learning ideas. The resulting scheduling algorithm can be run online, and is model-free with respect to the wireless channel parameters. Numerical results have demonstrated that our approach consistently and significantly outperforms other suboptimal sensor scheduling policies. 
 Future work will include the study of model-based reinforcement learning techniques \cite{Pong_temporal_difference}, to possibly improve the speed of learning when additional knowledge about the channel parameters is available. 


\appendix
\section{Proof of Theorem \ref{thm_stability}}
\label{sec:thm_stability_proof}
In the case $\rho_{\textnormal{max}} < 1$, condition (\ref{stability_condition}) is always satisfied. Indeed, in this case each process is stable and so the MDP (\ref{sensor_scheduling_MDP}) has  bounded average cost even when there are no sensor transmissions. 

Thus we concentrate on the case $\rho_{\textnormal{max}} \geq 1$.
Let $$m^* \triangleq  \textnormal{arg}\!\!\!\!\max_{m=1,\dots,M} q_m.$$ First assume a single channel system where only channel $m^*$ is available.
Consider a suboptimal policy where at each time instant the sensor with the largest holding time is chosen to transmit, provided that this holding time is greater than some $L > 2N$ \cite{MesquitaHespanhaNair}.  Using an argument similar to the proof of the first part of Theorem~3 in \cite{MesquitaHespanhaNair}, we can show that  this policy has bounded average cost if
\begin{equation}
\label{Mesquita_condition}
\rho_{\textnormal{max}}^2 P_L^{1/L} < 1,
\end{equation}
where $P_L$ can be expressed as $$P_L = \sum_{n < N} \mathbb{P}(n \textnormal{ successful transmissions in } L \textnormal{ time steps}).$$

The rest of the argument in Theorem 3 of \cite{MesquitaHespanhaNair} assumes i.i.d. packet dropping channels. To extend the argument to Markovian packet drops as considered in the current work, we make the following observation: 
Given that there are $n$ successful transmissions, then there will be $L-n$ failed transmissions in these $L$ time steps. Of these $L-n$ failed transmissions, \emph{at most $n$ of them will have followed a successful transmission} (or equivalently at least $L-2n$ of them will have followed a failed transmission). From this observation, we have
\begin{equation}
\label{P_L_bound}
\begin{split}
P_L & = \sum_{n < N} \mathbb{P}(n \textnormal{ successful transmissions in } L \textnormal{ time steps}) \\
   & \leq \sum_{n < N} \binom{L}{n} (\max(q_{m^*},1-p_{m^*}))^n   \\ & \quad \quad \times (\max(p_{m^*},1-q_{m^*}))^n (1-q_{m^*})^{L-2n} \\
   & \leq (N-1) \binom{L}{N-1}  (1-q_{m^*})^{L-2n}.
\end{split}
\end{equation}
In the first inequality in (\ref{P_L_bound}), the term $(\max(q_{m^*},1-p_{m^*}))^n$ upper bounds the probability of having $n$ successful transmissions, while the term $(\max(p_{m^*},1-q_{m^*}))^n   (1-q_{m^*})^{L-2n}$ upper bounds the probability of having $L-n$ failed transmissions, with at least $L-2n$ also having the previous transmission fail. 
The second inequality in (\ref{P_L_bound}) holds as $\binom{L}{n} \leq \binom{L}{N-1}$ for all $n < N$ if $L > 2N$. 
Taking limits in (\ref{P_L_bound}) gives
\begin{align*}
\lim_{L\rightarrow \infty} P_L^{1/L} & \leq  \lim_{L\rightarrow \infty}(N\!-\!1)^{1/L} \binom{L}{N\!-\!1}^{1/L} \!\!\! (1\!-\!q_{m^*})^{(L-2n)/L} \\ & = 1-q_{m^*}.
\end{align*}
Then by Assumption \ref{stability_assumption}, the condition (\ref{Mesquita_condition}) can always be satisfied for $L$ sufficiently large, and so the suboptimal policy has bounded average cost.
Thus the MDP (\ref{sensor_scheduling_MDP}) with only the single channel $m^*$ has bounded optimal average cost. As utilizing additional channels does not increase the optimal average cost, the result follows. 

\bibliography{IEEEabrv,sensor_scheduling}
\bibliographystyle{IEEEtran} 


\end{document}